\documentclass[a4paper, 11pt]{article}

\usepackage{geometry}

\usepackage{amsfonts}
\usepackage{amssymb}
\usepackage{amsmath}
\usepackage{amsthm}
\usepackage{verbatim}
\usepackage{enumitem}
\usepackage{xcolor}
\usepackage{tikz}
\usepackage{mathdots}
\usepackage{graphicx}
\usepackage{rotating}
\usetikzlibrary{calc,decorations.pathreplacing}

\theoremstyle{plain}
\newtheorem{theorem}{Theorem}[section]

\theoremstyle{definition}

\usepackage{doc}
\title{On expressive rule-based logics}

\author{
Antti Kuusisto\\
Tampere University 
}

\date{}

\begin{document}

\maketitle


\vspace{0.5cm}

\begin{abstract}
\noindent
We investigate a family of rule-based logics.
The focus is on very expressive languages. We provide a range of characterization 
results for the expressive powers of the logics and relate them with 
corresponding game systems.
\end{abstract}

\vspace{1.5cm}

\section{Introduction}

In this article we introduce and investigate very 
powerful logics based on rules in the style of Datalog 
(see, e.g., \cite{ebbinghaus}, \cite{libkin}, \cite{datal})
and Prolog. The point is then to couple the related 
languages with the framework developed in \cite{gamesandcomputation}. We
also investigate the Turing-complete logic  
defined in \cite{turingcomp}, \cite{gamesandcomputation} and
based on game-theoretic semantics. In fact, many of the 
results obtained below have counterparts in the setting of the 
Turing-complete logic of \cite{turingcomp}, \cite{gamesandcomputation}.

We begin the story by a recap of \emph{systems} as
defined in \cite{gamesandcomputation}. We
then define a rule-based logic $\mathrm{RLO}$ 
which is tailor-made for classifying finite 
ordered structures. We show how to 
capture $\mathrm{RE}$ with $\mathrm{RLO}$.
The logic $\mathrm{RLO}$ is rule-based, with, inter alia, Datalog-style
rules and beyond. Computations with $\mathrm{RLO}$ are deterministic. We
then lift the restriction to ordered structures and 
investigate $\mathrm{RL}$ which we show to capture $\mathrm{RE}$ 
without the assumption of models having a distinguished linear order.
The next step is to consider systems with nondeterministic rules.
To this end, we define $\mathrm{NRL}$. As an extension of $\mathrm{RL}$ it 
trivially captures $\mathrm{RE}$, but we show a somewhat stronger result
relating to model constructions. 
We also establish an analogous result for a version of
the Turing-complete logic from \cite{gamesandcomputation}. In 
fact, a rather similar result has
already been established in \cite{gamesandcomputation}.
We then investigate $\mathrm{GRL}$ which is tailor-made for 
\emph{systems} as defined in \cite{gamesandcomputation}.

Concerning the logic $\mathrm{RL}$ and its many variants we study,
there exist various languages with essentially the same 
model recognition capacity.
These include, inter alia, the while languages discussed in \cite{vianu} and
the Turing complete logic of \cite{turingcomp}.
However, $\mathrm{RL}$ and its variants have quite nice qualities, relating 
especially to simplicity and flexibility of use. Notably, the
logics $\mathrm{NRL}$ and $\mathrm{GRL}$ offer
various interesting features for modeling
scenarious, thereby having useful properties 
that go beyond mere recognition. Furthermore, compared to the 
Turing-complete logic of \cite{turingcomp}, the variants of $\mathrm{RL}$ are 
inductive whereas the logic of \cite{turingcomp} is coinductive.\footnote{However, it is not
difficult to simulate inductive computations rather directly in the
logics of \cite{turingcomp},\cite{gamesandcomputation} simply by considering the 
corresponding computation tables (or game graphs of the 
semantic game). And also the reverse 
simulation of the logics in \cite{turingcomp},\cite{gamesandcomputation} is possible.}

\section{Preliminaries}

We denote models by letters of type $\mathfrak{A}$, $\mathfrak{B}$, $\mathfrak{M}$, and so on.
The domain of $\mathfrak{A}$ is denoted by $A$ and similarly for the other letters. For 
simplicity, we sometimes write $R$ to indicate both the relation $R^{\mathfrak{A}}$
and the relation symbol $R$.

For simplicity, models are assumed to have a relational vocabulary (no function or constant symbols).
Also, they are assumed finite (with also a finite vocabulary),
although it will be easy to see that many of the results below do not really depend upon
this assumption. Also, the exclusion of function and constant symbols could be easily 
avoided by considering partial function symbols. We omit this option indeed for the sake of
simplicity. We assume there exists a canonical linear ordering $<_{\mathit{symb}}$ of the
full infinite set of all relation symbols. This enables unique binary encodings of models.

The \textbf{encoding} of a model $\mathfrak{M}$ with respect to a linear ordering
ordering $<$ of $M$ is the binary string $\mathit{enc}_<(\mathfrak{M})$ 
defined such that it begings with $|M|$ bits $1$ followed by a single $0$, and
after this are the encodings of the relations as follows.
\begin{enumerate}
\item
The relations are encoded as a concatenation of all the relation encodings, 
one relation at a time, in the order 
indicated by $<_{\mathit{symb}}$.
\item
For a $k$-ary relation $R^{\mathfrak{M}}$, we 
simply list a bit string of length $M^{k}$ 
where the $n$th bit is $1$ iff
the $n$th tuple (with respect to the standard lexicographic order of $M^k$
defined with respect to the linear ordering $<$ of $M$) is in
the relation $R^{\mathfrak{M}}$.
\end{enumerate}
We note that this encoding scheme is similar to the one defined in \cite{libkin}. 
We also note that obviously $<$ does not necessarily need to be in the vocabulary of the
model to be encoded.
We may write $\mathit{enc}(\mathfrak{M})$ when the linear ordering $<$ is 
known from the context of irrelevant. It is of course obvious that different 
linear orderings of the model domain are bound to give different encodings.

The logic $\mathcal{L}_{\mathrm{RE}}$ \cite{sofsem} consists of
sentences of the form $IY \exists X_1\dots \exists X_n\psi$
where the part $\exists X_1\dots \exists X_n\psi$ is a formula of existential second order logic
and $IY$ a new operator (where $Y$ is a unary second-order relation variable).
Here $\psi$ is the first-order part. 
We have $\mathfrak{M}\models IY \exists X_1\dots \exists X_n\psi$ if we can
expand the domain of $\mathfrak{M}$ by a finite set $S$ of new 
elements such that $$(\mathfrak{M}+S, Y\mapsto S)\models \exists X_1\dots \exists X_n\psi$$
where $(\mathfrak{M}+S,Y\mapsto S)$ is the model obtained
from $\mathfrak{M}$ by the following operations.
\begin{enumerate}
\item
We first extend the domain of $\mathfrak{M}$ by the set $S$ of new elements.
The relations are kept as they are.
\item
We then expand the so obtained model to interpret the new unary 
symbol $Y$ as the set $S$. (That is, $Y$ names the fresh elements in the domain.)
\end{enumerate}

Let $\tau$ be a vocabulary and consider a class $\mathcal{C}$ of finite $\tau$-models. We say 
that a Turing machine $\mathit{TM}$ defines a semi decision
procedure for $\mathcal{C}$ if $\mathit{TM}$ accepts a bit string $s$
iff $s = \mathit{enc}_<(\mathfrak{M})$ for some $\mathfrak{M}\in\mathcal{C}$ and
some linear ordering $<$ of the 
domain of $\mathfrak{M}$. When not accepting, the machine does not have to halt.
A model class is in $\mathrm{RE}$ (\emph{recursively enumerable}) iff there is a Turing 
machine that defines a semi decision procedure for it.
When considering ordered models, i.e., models where some distinguished
predicate $<$ in the vocabulary is
always a linear ordering of the domain, we can---even then---use the 
above definition for semi decision procedures for model classes. However, we can then also use
the following clearly equivalent definition: a Turing machine $\mathit{TM}$ defines a semi decision
procedure for $\mathcal{C}$ if $\mathit{TM}$ accepts a bit string $s$
iff $s = \mathit{enc}_<(\mathfrak{M})$ for some $\mathfrak{M}\in\mathcal{C}$ and the
distinguished ordering $<$ of the 
domain of $\mathfrak{M}$. Both definitions result in the same class of semi decidable classes of ordered models.

\section{On systems}

\subsection{Elements of systems}

We now consider $\textbf{systems}$ as defined in \cite{gamesandcomputation}.
Let $\sigma$ be a signature and $A$ a set of actions. Let $I$ be a set of agents.
(Technically $A$ and $I$ are simply sets.) Consider a set $S$ of structures
over the vocabulary $\sigma$. We note that in one interesting
and significant case, $\sigma$ has only unary relation symbols and $S$ is simply a
set of states (or points with some local information based on 
unary predicates, i.e., propositional valuations). Then we will ultimately end up with just a
slight generalization of Kripke models.
However, it is also instructive to think of the structures in $S$ 
simply as relational first-order $\sigma$-models in the usual sense of model theory.

Let $\mathcal{T}$ denote the set of $(S,A,I)$-sequences; as
defined in \cite{gamesandcomputation}, an $(S,A,I)$-\textbf{sequence} is a
finite tuple \[(\mathfrak{M}_0,a_0,\dots , \mathfrak{M}_n,a_n)\] where 
each $\mathfrak{M}_i$ is a structure in $S$ and $a_i$ is a tuple of
actions in $A^I$.  
We note that the following generalizations could be possible (but we
omit considering them here explicitly).
\begin{enumerate}
\item
Instead of letting $a_i\in A^I$ be a tuple of actions involving
any individual actions from $A$, we can define a function that limits the available actions based on
the earlier sequence, meaning that a set $A[i]\subseteq A$ can be
determined by $(\mathfrak{M}_0,a_0,\dots , \mathfrak{M}_{i})$
and then we must have $a_i \in (A[i])^I$.
\item
Furthermore, we can let the set of active agents be $I[i]\subseteq I$, similarly 
determined by $(\mathfrak{M}_0,a_0,\dots , \mathfrak{M}_{i})$. Then we
must have $a_i\in A[i]^{I[i]}$.
\item
Yet further, we can make the actions available to each individual agent depend also on
the earlier sequence. Formalizing all this is a triviality. Indeed, we
must then have $a_i\in A[i]^{I[i]}$ with the 
additional condition that the $k$-th member of $a_i$ (the action of the $k$-th agent in $a_i$)
must be chosen from $A[i][k]\subseteq A$ of 
actions available to agent $k$ in round $i$, with $A[i][k]$ being 
determined by $(\mathfrak{M}_0,a_0,\dots , \mathfrak{M}_{i})$.
\end{enumerate}
Note that also the empty sequence is a $(S,A,I)$-sequence.

A \textbf{system frame base}, as
given in Definition 3.1 of \cite{gamesandcomputation}, is a pair $(S,F)$ where $S$ is a
set of $\sigma$-structures and $F$ is a function $F: T\rightarrow \mathcal{P}(S)$ where $T$ is
some subset of the set of all
sequences in $\mathcal{T}$ (where $\mathcal{T}$ denotes the set of all $(S,A,I)$-sequences).

According to Definition 3.2 in
\cite{gamesandcomputation}, a \textbf{system frame} is a tripe $(S,F,G)$
where $(S,F)$ is a system frame base and $G$ is a function 
mapping from some set $E\subseteq \mathcal{T}\times \mathcal{P}(S)$
into $S\cup \{\textbf{end}\}$ such that when $G((x,U))\not= \textbf{end}$, the
condition $G((x,U))\in U$ holds.\footnote{Note here that $G$ is undefined or outputs \textbf{end}
when the set $U$ is the empty set.} Thus $G$ is essentially a
choice funtion that \emph{chooses} the
\emph{actual next world} from the set of 
\emph{possible future worlds} chosen by $F$. It can be interpreted, e.g., as \emph{chance} or some kind of a grand controller of the system.

Finally, a \textbf{system}, as given in Definition 3.3 of \cite{gamesandcomputation}, is
defined as a
tuple $$(S,F,G,(f_i)_{i\in I})$$ where $(S,F,G)$ is a 
system frame and each $f_i$ is a function $f: V_i\rightarrow A$
with $V_i$ being some subset of the set $\mathcal{T}'$ of tuples
$$(\mathfrak{M}_0,a_0,\dots , \mathfrak{M}_k),$$
that is, tuples that are like $(S,A,I)$-sequences but with the last tuple of
actions (which in the
above tuple would be $a_k$) removed. Intuitively, $f_i$ is a 
strategy that gives an action for the agent $i$ based on
sequences of the above mentioned type, i.e., the type
$$(\mathfrak{M}_0,a_0,\dots , \mathfrak{M}_k).$$
In fact, as in \cite{gamesandcomputation}, we can even identify $f_i$ with
the agent. The agent is the strategy the agent follows.
For conveniences, let us call sequences of type $$(\mathfrak{M}_0,a_0,\dots , \mathfrak{M}_k)$$
\textbf{structure ended} $(S,A,I)$-\textbf{sequences}, or simply \textbf{structure ended sequences} 
when $(S,A,I)$ is clear from the context or irrelevant. We note that $S$ is 
called the \textbf{domain} of the system, and it should not be confused by 
the domains of individual models in $S$.

Systems evolve as given in \cite{gamesandcomputation}. However, the next section
defines some scenarios with a closer look at constraints on system evolution.

\subsection{Controlling systems}

As discussed in \cite{gamesandcomputation}, it is interesting to consider a 
framework where the agents do not see the structures $\mathfrak{M} \in S$ 
directly, but a perception of them. In that article, this is realized by 
defining two functions $p_i$ and $d_i$ for each agent $i\in I$. There are
many ways to define $p_i$ and $d_i$. 
Here we define the function $p_i$ so that it maps from the set of nonempty
structure ended sequences in the underlying system to a new set $P_i$ of models;
the models in $P_i$ can be considered, e.g., as 
\textbf{perceived} (or perceivable) 
\textbf{models} for the agent $i$. Thus $p_i$ is a some kind of a
perception function that gives the current perceived model to an agent. The signature of 
the models in $P_i$ need not be similar to the signature of the models in the domain $S$ of the
underlying system.\footnote{Often the models $P_i$ can be just propositional valuations. However, the same also applies to the domain $S$ of the system, and the scenario where both perceived models and the models in $S$ are propositional valuations is of course important.}

The function $d_i$ maps
from a set $M_i\times \langle P_i \rangle$ into the set $A_i\times M_i$. Here $M_i$ is
simply a set that
contains the \textbf{mental states} of the agent $i$ and $\langle P_i\rangle$ is 
the set of finite nonempty sequences
$$(\mathfrak{P}_0,a_0,\dots , \mathfrak{P}_{k-1},
a_{k-1},\mathfrak{P}_k)$$
that contain perceived models $\mathfrak{P}_j\in P_i$ of the agent $i$ and
action tuples $a_j\in A^I$ by all the agents (so $P_i$ is indeed the set of perceivable 
models of agent $i$). And $A_i$ is the set of actions available to agent $i$, so in the
general case, $A_i = A$. Note that these 
sequences end with a structure, so
they are quite similar to the structure ended sequences
$$(\mathfrak{M}_0,a_0,\dots , \mathfrak{M}_{k-1},
a_{k-1},\mathfrak{M}_k)$$
of the system itself. However, in the
perception sequences we only have the perceived rather than real models (the 
real models are the models in the system domain $S$).
Note that it is very natural to limit $d_i$ so that it does not 
depend on the actions of agents other than the agent $i$. This 
means that each $a_j$ in a perception sequence is replaced by the
single action $a_j(i)$ by agent $i$ in the tuple $a_j\in A^I$.
Then information (of varying quality) about the actions of other agents in the 
perception sequences can be considered to be encoded, for example, in (if anywhere) the 
next perceived model $\mathfrak{P}_{j+1}$.
Now, it is highly natural to make $d_i$ depend only on the current 
mental state and the last perceived model of $i$. We call 
such a $d_i$ mapping from $M_i\times P_i$ into $A_i\times M_i$ a \textbf{simple} $d_i$.
Whether we use a simple $d_i$ or not, the system evolves as follows.

First we have a system and the functions $d_i$ and $p_i$ for each agent $i$. We
also have an initial mental state $m_{\mathit{initial}}(i)\ \in\ M_i$ for each agent $i$.
The system itself determines an initial structure $\mathfrak{M}_0\in S$ (or alternatively, we
arbitrarily just appoint the structure $\mathfrak{M}_0$).
Inductively, from any structure ended sequence $$(\mathfrak{M}_0,a_0,\dots , \mathfrak{M}_k)$$
we obtain, using
the function $p_i$ for each agent $i$, the next perceived model $\mathfrak{P}_i$ for the agent $i$.
From there, we use $d_i$ to determine the new mental state $m\in M_i$ of
the agent and the action $a_i'\in A_i$ by the agent. For simple $d_i$, we
have $d_i:M_i\times P_i\rightarrow A_i\times M_i$ and for a
non-simple one $d_i:M_i\times \langle P_i\rangle \rightarrow A_i\times M_i$.
(There are of course many relevant variants between simple and general non-simple, e.g.,
taking into account the full tuple of immediately previous actions by the agents.)
With the mental state of each agent $i$ updated, and with an action $a_i'\in A_i$ 
for each agent determined, we do the following. We build the tuple $a_k\in A^I$ of
actions by all agents from the (at this stage known) individual actions $a_i'$ for each agent $i$.
Then we add this tuple $a_k$ to the structure ended sequence $$(\mathfrak{M}_0,a_0,\dots , \mathfrak{M}_k)$$ from where we started the description of the inductive step. Then, based on 
$$(\mathfrak{M}_0,a_0,\dots , \mathfrak{M}_k,a_k),$$
the system determines (using $F$ and $G$) the new model $\mathfrak{M}_{k+1}$.
Then we of course repeat the step similarly from the new structure ended 
sequence $$(\mathfrak{M}_0,a_0,\dots , \mathfrak{M}_{k+1}).$$

A particularly interesting setting (let us call it \textbf{M-finitary}) is the one 
where each $M_i$ is finite. Also the case (call it \textbf{MP-finitary}) where 
each $M_i$ and each $P_i$ is finite is interesting. The case where each $M_i$ and $P_i$ and also 
the domain $S$ of the system is finite can be called \textbf{MPS-finitary}.
The \textbf{p-quasi-finitary} case is the one where the range of each $p_i$ is 
finite (while $S$ need not be); the range of $p_i$ being finite means that $p_i$ is
intuitively finitary in the sense that it sees only finitely many (intuitively different)
cases that it maps differently. That is, the domain of $p_i$
partitions into finitely many equivalence 
classes (each class sharing an output) and $i$ sees the inputs in each class as being similar to
each other or even indistinguishable from each other.
We can even redefine the domain of each $p_i$ to consist of, e.g., finitely many isomorphism
classes of some relation $R_i$ interpreted by the models in $S$. The idea is
that $R_i$ represents the sphere of (menta, physical, or a combination of those) 
perception of the agent $i$. 
The further requirement of $R_i$ having only finitely many tuples could also perhaps be forced.

All in all, at least the following cases should be specially distinguished.
The first one is the MP-finitary case 
where each $d_i$ is simple (in the formal sense defined above), and furthermore, the set $I$ of
agents is also finite. Let us call the case \textbf{1-elementary}.
Note that in the $1$-elementary case, we can 
always assume that $A$ is finite (as the union of the ranges of the functions $d_i$ is finite). 
The 1-elementary which is also $p$-quasi-finitary is highly interesting (let us
call it \textbf{$2$-elementary}). In the $2$-elementary case, if each $p_i$
furthermore depends only on (i.e., the output is always determined by) the current 
model $\mathfrak{M}$ in the system domain $S$, we
call the case \textbf{elementary} and each $p_i$ is called an  
\textbf{elementary perception function}. Obviously
then we can simply regard $S$ as the domain of the functions $p_i$.
Finally, the elementary case where $S$ is finite can be called \textbf{strongly elementary}. 
The $\textbf{a-elementary}$ case is the $2$-elementary case where each $p_i$ depends on
the current model and the previous tuple of
actions.\footnote{Of course in the very beginning, there is no
previous tuple of actions, but in the subsequent rounds there is.}
The \textbf{strongly a-elementary} case is the a-elementary case with
the domain $S$ being finite.
The following sections contain logics for many different scenarios of
system simulation with ideas
visioned already in \cite{double}.


\section{Rule-based logics}

In this section we consider rule-based logics. While there are 
similarities to systems such as, e.g., Datalog variants, there are also various notable 
differences. We begin by considering ordered models.

\subsection{Ordered models}

Let $\tau$ be a relational vocabulary that contains a distinguished 
binary relation symbol $<$ which will always be considered a linear order
over the domain of the model investigated. We exclude constant symbols and 
function symbols from the vocabulary for the sake of simplicity. They could be
added however, especially if considering partial function symbols (noting also that
partial constant symbols would be interpreted as constants that have at most one
reference point in the model domain). Nevertheless, we indeed let $\tau$ be a
relational vocabulary here to streamline the exposition. We note that $\tau$ is assumed to be
finite (although it will be trivial to see which results 
would go through for infinite $\tau$). The 
vocabulary $\tau$ can contain nullary relation symbols. Recall that a 
nullary symbol $Q$ is interpreted either as $\top$ or $\bot$ (true or false) by a model 
with $Q$ in the vocabulary.

Let $\tau^+$ be an \textbf{extended vocabulary}, $\tau^+ \supseteq \tau$, 
where the part $\tau^+ \setminus \tau$ contains ``relation symbols''
dubbed \textbf{tape predicates}. These are exactly as relation symbols but
they are not considered to be part of the underlying vocabulary $\tau$. Tape 
predicates can be nullary. In the beginning of computation, tape
predicates are interpreted as the empty relation (and each
nullary tape predicate as $\bot$). On the technical level, we shall mostly try to
reserve the terms \emph{relation symbol} and \emph{tape predicate} for different and
disjoint sets of
symbols. Tape predicates are auxiliary and relation symbols part of the input model.
However, we can of course define models that interpret tape predicates as if they were relation symbols.

A \textbf{transformation rule of the first kind} is a construct of the form
\[X(x_1,\dots , x_k)\ :-\ \ \varphi(x_1,\dots , x_k) \]
where $X\in \tau^+$ is a $k$-ary symbol and $\varphi(x_1,\dots , x_k)$ is a
first-order $\tau^+$-formula\footnote{The set of
relation symbols and tape predicates in a $\tau^+$-formula is required to be
some subset of $\tau^+$. A $\tau^+$-formula may also be called a formula in 
the vocabulary $\tau^+$.} whose set of free variables is
precisely $\{x_1,\dots , x_k\}$. We note, expecially with the reader 
familiar with Datalog and similar languages in mind, that $X$ can but does not have to be a 
tape predicate, it can be any symbol in $\tau^+$. We note also, concerning the 
variables $x_1,\dots , x_k$, that the variables do not have to be pairwise distinct.
The symbol $X$ is called the $\textbf{head}$ symbol of the rule, 
and the left-hand side formula $X(x_1,\dots , x_1)$ simply the \textbf{head} of
the rule. The right-hand side formula $\varphi(x_1,\dots , x_k)$ is
the \textbf{body} of the rule.
Transformation rules of
the first kind will also be called \textbf{$1$-transformers}. We stress that the
formula $\varphi(x_1,\dots ,x_k)$ can indeed be any first-order formula within the
given constraints: it does not have to be free or negations or quantifiers or anything like that.

Now, let $\mathfrak{M}$ be a $\tau^+$-model.\footnote{The
set of relation symbols and tape predicates
that a $\tau^+$-model interpretes is precisely $\tau^+$.} Consider a $1$-transformer $F$ of
the form \[X(x_1,\dots , x_k)\ :-\ \ \varphi(x_1,\dots , x_k).\]
We let $\mathcal{F}$ be the operator such that 
$$\mathcal{F}(F,\mathfrak{A})\ =\  
\{\, (a_1,\dots , a_k)\in A^k\, |\, 
\mathfrak{A}, (x_1\mapsto a_1,
\dots , x_k\mapsto a_k)\models \varphi(x_1,\dots , x_k)\ \}$$
where $(x_1\mapsto a_1,
\dots , x_k\mapsto a_k)$ is the assigment that maps $x_i$ to $a_i$ for each $i\in \{1,\dots , k\}$.
Therefore, $\mathcal{F}$ is the operator that takes any $1$-transformer $F$ and
model $\mathfrak{A}$ (where the symbols in $F$ are in the vocabulary) as an 
input and gives the relation determined by the rule body as the output. To put this shortly,
the operator $\mathcal{F}$ evaluates the rule $F$ on the input model; the evaluation process is
similar to the one in, e.g., Datalog. As in Datalog, we
can use $\mathcal{F}$ to update $\mathfrak{A}$ by replacing the relation $X^{\mathfrak{A}}$ 
corresponding to the head symbol by the relation $\mathcal{F}(F,\mathfrak{A})$.

A \textbf{transformation rule of the second kind}, or a \textbf{$2$-transformer}, is the
construct denoted by
\[I\]
which simply adds a single domain element to the current model and
extends the relation $<$ so that the new element becomes the 
last element in the order.
Other relations are kept as they are. A \textbf{conditional $2$-transformer} is a
rule of the form 
\[ I\ :-\ \ \varphi\]
where $\varphi$ is a first-order sentence in the vocabulary $\tau^+$.
The interpretation is that if $\varphi$ holds, then we extend 
the domain (in the same way as the rule $I$ does), and
otherwise we do not extend the domain, we just move on. A \textbf{transformation 
rule of the third kind}, or a \textbf{$3$-transformer}, is a rule of the form
\[ D\ :-\ \ \psi(x) \]
where $\psi(x)$ is a first-order formula in the
vocabulary $\tau^+$ and with a single free variable, $x$.
The rule deletes from the model $\mathfrak{A}$ precisely all elements $a_0$
such that we have $\mathfrak{A}, (x\mapsto a_0) \models \psi(x)$. One can 
also easily define natural conditional $1$-transformers and $3$-transformers, 
with the idea that whether or not they are executed depends on an
additional first-order sentence.
If the first-order sentence holds, we execute the
rule, and if the first-order sentence fails to hold, we just move on without changing the model.

A \textbf{control rule} is a rule of the form 
\[ k \]
where $k$ is a positive integer written in binary. The rule asserts that we
should go execute the rule number $k$ if such rule exists (i.e., we jump to
the rule $k$ if possible). If such rule
does not exist, the computation halts. We shall define later on how rule
numbers are used exactly. A \textbf{conditional control rule} is a rule of the form
\[ k\ :-\ \ \chi \]
where $k$ is a positive integer written in binary and $\chi$ is a
first-order sentence in the vocabulary $\tau^+$. The rule states that if $\chi$
holds, then we should go to execute the rule number $k$. If $\chi$ does not hold, we move on to the
next rule (and if there is no next rule, the computation halts). If
there is no rule $k$ at all but $\chi$ holds, then the computation halts.

A \textbf{program} is a finite
sequence of rules, i.e., a
list of the form
\begin{align*}
1:&\empty{ }\ F_1\\
\vdots \\
k: & \empty{ }\ F_k
\end{align*}
where each $F_i$ is a rule (please see an example in
the proof of Theorem \ref{rlotheorem} below). We begin each 
\textbf{line} with a number (the rule number) and a colon.
The rule number---officially written in binary---helps in using control rules.
However, the rule number and colon can of course be dropped, as they 
increase in the obvious way,
beginning with $1$. The program is executed one rule at a time (unlike
Datalog), starting from rule $1$ and proceeding from there:
if we are at rule $F$, and it is not a control rule, we first execute the rule and
then move on to the next rule below the current one. Also, if the last rule is
executed (and it is not a control rule leading to a jump to some existing rule), then the
computation ends after that.
Control rules allow for jumps that do not necessarily proceed in the 
way indicated by the rule numbers.
If a control rule leads to a rule number (i.e., line number) for
which there is no rule, the computation ends.
Recall that the transformer rules transform the model in
the way described above, so the model typically changes as the computation progresses.

Consider a program $\Pi$ where the set of relation symbols is $\sigma_0$ and
tape predicates $\sigma_1$ (obviously $\sigma_0\cap
\sigma_1 = \emptyset$). Let $\mathfrak{M}$ be a model
whose vocabulary is $\tau\supseteq \sigma_0$. (Note
that $\tau\cap\sigma_1 = \emptyset$.) Computation with the input $\mathfrak{M}$ then 
proceeds as described above, starting with the 
expansion of $\mathfrak{M}$ to the vocabulary $\tau^+ = \tau \cup \sigma_1$
such that tape predicates are interpreted as empty relations (and $\bot$ for
nullary tape predicates). We call this expansion the $\Pi$\textbf{-expansion of $\mathfrak{M}$}.

Now, consider a scenario where $\sigma_1$ contains
the nullary tape predicate $X_{true}$. We consider this a special tape predicate and
write $\mathfrak{M}\models \Pi$ for a model of vocabulary $\tau$
and a program of vocabulary $\tau^+ = \tau\cup \sigma_1$ if the computation
beginning with the the $\Pi$-expansion of $\mathfrak{M}$ ultimately 
halts such that $X_{true}$ holds (is equal to $\top$) in the final model at halting.

We call $\mathrm{RLO}$ (for \emph{rule logic with order}) the system 
consisting of programs with $1$-transformers, $2$-transformers and 
conditional control rules, as described above.
Conditional transformer rules are not included. Note that non-conditional 
control rules can of course be simulated with conditional control rules.
We say that $\mathrm{RLO}$ \textbf{defines} a
class of $\tau$-models $\mathcal{C}$ if there is a program $\Pi$ 
such that for all $\tau$-models $\mathfrak{M}$, we 
have $\mathfrak{M}\models\Pi$ iff $\mathfrak{M}\in \mathcal{C}$. Note here
that we of course restrict attention to finite models only.
The same definition of definability also applies to other logics we
shall consider in this article.

\begin{theorem}\label{rlotheorem}
Let $\tau$ be a vocabulary with $<$ and limit attention to ordered $\tau$-models.
Then, $\mathrm{RLO}$ can define a class $\mathcal{C}$ of 
$\tau$-models iff $\mathcal{C}$ is in $\mathrm{RE}$.
\end{theorem}

\begin{proof}
It is clear that computations with $\mathrm{RLO}$ can be simulated by a
Turing machine. For the other direction, suppose a Turing machine $\mathit{TM}$ 
recognizes some class $\mathcal{C}$ of ordered $\tau$-models. Now recall from the 
preliminaries the logic $\mathcal{L}_{\mathrm{RE}}$ that can define precisely the
semi-decidable classes of models. Therefore 
there is a sentence $\mathrm{I}Y\exists X_1\dots \exists X_n\varphi$ of $\mathcal{L}_{\mathrm{RE}}$
that defines $\mathcal{C}$ with respect to the class of all finite ordered $\tau$-models.
Now notice that for any formula $\exists Z_1\dots \exists Z_m\psi$ of 
existential second-order logic (with $\psi$ being first-order), there clearly exists an equivalent 
formula $\exists Z \beta$ (with $\beta$ being first-order) where the 
arity of $Z$ is the sum of the arities of the predicates $Z_1,\dots , Z_m$.
Now, let $\exists Z\psi'$ (with $\psi'$ being first-order) be an existential
second-order the formula 
equivalent to $\exists X_1\dots \exists X_n\varphi$, and  
suppose that the sum of the arities of the 
relation variables $X_1,\dots , X_n$ is $k$. Thereby the arity of $Z$ is $k$. 
Note that now $\mathrm{I}Y \exists Z\psi'$ is equivalent to the formula $\mathrm{I}Y\exists X_1\dots \exists X_n\varphi$ of $\mathcal{L}_{\mathrm{RE}}$
that defines $\mathcal{C}$. We will 
next write a program that is equivalent to $IY\exists Z \psi'$. Note that $Y$ and $Z$ will of
course be tape predicates, and the program will also use some other tape predicates.

Let $\mathit{step}_{Z}(x_1,\dots , x_k)$ be the first-order formula whose interpretation
over any ordered model $\mathfrak{N}$ (which interprets $Z$) is the relation $R\subseteq N^k$ 
such that the binary encoding of $R$ is the string that points to the integer $z$
that is one larger than the integer $z'$ that $Z^{\mathfrak{N}}$ points to. If $z'$ is
already the maximum integer, then $R$ will simply be equal to $Z^{\mathfrak{N}}$.
The binary encoding here for relations is of course the one described in the preliminaries.
The formula $\mathit{step}_Z(x_1,\dots ,x_k)$ is routine to write using $<$.

Next, let $\mathit{max}(Z)$ be the first-order formula which states
that $Z$ is the total relation over the current domain. Note that this is
equivalent to the binary encoding of $Z$ being the maximum 
string (containing only bits $1$) with respect to all bit strings of
length $d^k$ where $d$ is the size of the current domain.
The required program is as follows.
Note that for readability, none of the indices are in binary.
\[
\begin{array}{rll}
1:&\empty{ }\ \ X_{true}\ &:-\ \ \psi'\\
2:&\empty{ }\ \ 365\        &:-\ \ X_{true}\\
3:&\empty{ }\ \ 6\        &:-\ \ \mathit{max}(Z)\\
4:&\empty{ }\ \ Z(x_1,\dots , x_k)\        &:-\ \ \mathit{step}_Z(x_1,\dots , x_k)\\
5:&\empty{ }\ \          & 1\\
6:&\empty{ }\ \ \        & I\\ 
7:&\empty{ }\ \ Y(x)\        &:-\ \ Y(x) \vee \neg\exists y(x<y)\\
8:&\empty{ }\ \ Z(x_1,\dots , x_k)\        &:-\ \ \bot\\
9:&\empty{ }\ \ \        & 1\\ 
\end{array}
\]
The program tests if $\psi'$ holds, and if not, it modifies $Z$ to be
the next relation with respect to the lexicographic ordering of $k$-tuples
defined with respect to $<$. Once all relations $Z$ have been tested, the 
domain is extended and $Z$ is set to be the
empty relation. The procedure is then repeated.
\end{proof}

Note that we can clearly add all the conditional rules and $3$-transformers (and also a
conditional $3$-transformer), and 
the resulting system will still define precisely the classes of
ordered models in $\mathrm{RE}$. Indeed, we only need to prove that the
stronger system can be simulated by a Turing machine, and 
this is clear. And we can do even more, of course.

A \textbf{conditional rule tuple} is a construct of the form
$$(\mathit{If}\, \varphi_1\, \mathit{then}\, F_1,\, 
\mathit{else\, if}\, \varphi_2\, \mathit{then}\, F_2,\ \dots\ ,\,  
\mathit{else\, if}\, \varphi_{k-1}\, \mathit{then}\, F_{k-1},\ \mathit{else}\, F_k)$$
where each $\varphi_i$ is a first-order $\tau^+$-sentence
and $F_i$ is a rule (of any kind discussed above), and we have $k\geq 1$, so 
singleton tuples are allowed.\footnote{A singleton 
tuple is just a rule $F_1$.} A conditional rule tuple
will occupy a line in a program just like the rules above. For example, if $C$ is a
conditional rule tuple, then for example the line $6$ in the program could be of type 
\[6:\, C\]
A conditional rule tuple is interpreted in the most 
obvious way as follows: take the first
precondition formula $\varphi_i$ that holds and then execute rule $F_i$. If no
precondition rule holds, execute the last rule $F_k$. After executing that one
rule $F_i$, do the following.
\begin{enumerate}
\item
If the rule $F_i$ instructed to jump to line $m$, continue 
from that line. If the line $m$ does not exist, the computation stops.
\item
If the rule $F_i$ did not instruct to jump, go to the next line after the
current conditional rule tuple (having of course executed $F_i$ already). If
that next line does not exist (meaning we are at
the last line of the program), the computation stops.
\end{enumerate}

A \textbf{parallel rule} is a tuple of the form $(G_1,\dots , G_k)$
where each $G_i$ is a conditional rule tuple.
A parallel rule is executed as follows (where we---at first---assume no
deadlocks arise).

First note that each conditional rule tuple $G_i$
determines one rule $F_i$ to be executed. 
These rules $F_1,\dots , F_k$ (one rule for each $G_i$) are executed in parallel as follows. We
first execute the transformation rules $F_i$ (simultaneously, based, on the
current $\tau^+$-model $\mathfrak{M}_{current}$)
that do not involve deleting or adding domain points (so no $I$ or $D$ in the syntax). Being
transformation rules, these
rules do not involve jumps either. This way we obtain a
model $\mathfrak{M}_1$. Then we do the rules $F_i$ with
deletions ($D$ appears in the rule), but not
based on $\mathfrak{M}_1$ but the model $\mathfrak{M}_{current}$ instead. We 
end up with the variant of $\mathfrak{M}_1$ that has the points to be deleted indeed
removed. Note that we remove the 
union of the points that the successful deletion rules instruct to be deleted, and if
there are conditional deletion rules, the condition is evaluated against $\mathfrak{M}_{current}$.
Then we do the additions (one point per successful addition rule; again if the 
rule is conditional, the condition is evaluated with respect to $\mathfrak{M}_{current}$).
The model after the additions is the new model $\mathfrak{M}_{new}$ that the parallel rule constructs.
If deadlocks arise at the above stages (meaning
that at least two rules would treat some head predicate $X$ differently), the 
computation simply halts without the model being modified at all, i.e., with $\mathfrak{M}_{current}$. Note that deletion and addition rules cannot lead to deadlocks.
Finally, after the now described modification step (if the computation did not
lead to a deadlock), we check for control rules in $(F_1,\dots ,F_k)$. If there are no
control rules in $(F_1,\dots ,F_k)$, we
continue from the next line after the 
parallel rule; if there is no next line, the computation ends with $\mathfrak{M}_{new}$
being the final model.
If there are control rules among $(F_1,\dots , F_k)$, we 
first compile a list $L$ of all line numbers where we
should jump, with conditional control rules evaluated 
based on $\mathfrak{M}_{current}$. If there are
different numbers in $L$, this is a
deadlock, and the computation
halts (in this case without the modifications, i.e., with $\mathfrak{M}_{current}$
being the final model).
If there is a single jump instruction in $L$, we jump to the corresponding line and 
continue from there. If that
line does not exist, the computation ends with the modified model $\mathfrak{M}_{new}$.

Note that we can very naturally define parallel rules based on
transformation rules only. Deadlocks can always be avoided by using 
different head symbols in each collection of possibly parallel actions.
By this we mean forbidding the use of the same symbol as a head symbol in
different simultaneous conditional rule tuples $G_i$. By
using different nullary head symbols in the parallel
actions, we can even directly simulate jump rules; such a parallel
rule is then followed by jump rules 
specifying how to jump based on the nullary predicates.

All the above rules can be added to $\mathrm{RLO}$ and we can 
still simulate the resulting logic with a deterministic Turing machine.

We next turn to
the case without order. Most of the ideas and notions will be carried over to the 
following subsection relatively directly.

\subsection{Without order}

Above we investigated the case of models with an order. However, an
order is not really required. Consider the syntax of $\mathrm{RLO}$
(without the assumption of order). Redefine $\mathrm{RLO}$ such
that where we previously had a first-order formula, we can now use a
formula of existential second-order logic.
Call the resulting logic $\mathrm{RL}$ (\emph{rule logic}). The programs are run as
those of $\mathrm{RLO}$. However, models do not have any distinguished order predicate in them.
Note that we can still investigate ordered models with $\mathrm{RL}$, but then $<$ does
not automatically update itself to a linear order when the model domain is 
extended. Instead, now $<$ is
treated as other predicates. Indeed, $I$ just adds a point, and no order relation is extended.

\begin{theorem}\label{freeofordering}
For any $\tau$, 
$\mathrm{RL}$ can define a class $\mathcal{C}$ of 
$\tau$-models iff $\mathcal{C}$ is in $\mathrm{RE}$.
\end{theorem}
\begin{proof}
Simulating $\mathrm{RL}$ with a Turing machine is easy. For the other
direction, we again use $\mathcal{L}_{\mathrm{RE}}$ which defines 
precisely all $\mathrm{RE}$-classes of models (whether or not a
distinguished order is present). We need to find a program 
equivalent to a sentence $IY \exists X_1\dots \exists X_n\, \psi$ where $\psi$ is first-order.
The following does that.
\[
\begin{array}{rll}
1:&\empty{ }\ \ X_{true}\ &:-\ \ \exists X_1\dots \exists X_n \psi\\
2:&\empty{ }\ \ 365\        &:-\ \ X_{true}\\
3:&\empty{ }\ \ X_{domain}(x)\ &:-\ \ x=x\\
4:&\empty{ }\ \ \        & I\\ 
5:&\empty{ }\ \ Y(x)\ &:-\ \ Y(x) \vee \neg X_{domain}(x)\\
6:&\empty{ }\ \ \        & 1\\ 
\end{array}
\]
\end{proof}

To add nondeterminism to $\mathrm{RL}$, we introduce the
following rules.
$$\exists X$$
where $X$ is a relation symbol or tape predicate (of any arity) and
$$\exists(k_1,\dots , k_n)$$
where $(k_1,\dots , k_n)$ is a tuple of positive integers.
The rule $\exists X$ is executed such that we nondeterministically 
choose an interpretation for $X$. (The old interpretation of $X$ is overridden.)
The rule $\exists(k_1,\dots , k_n)$ is executed such that we 
nondeterministically jump to one of the rule numbers $k_1,\dots , k_n$ (that is, we
jump to a line with the chosen number). More rigorously, we nondeterministically
choose one of $k_1,\dots , k_n$ and then attempt to jump to the line with that
number. If such a rule (i.e., line) exists, we
continue from there. If not, the computation ends. We
can also allow for the case $\exists()$ where $()$ is the empty tuple. This simply
terminates the computation.\footnote{Of course similar 
rules $\forall X$ and $\forall(k_1,\dots , k_n)$ could be
defined to allow for an obvious way to include alternation into the picture.}

We call $\mathrm{NRL}$ (where $N$ stands for \emph{nondeterminism}) the 
extension of $\mathrm{RL}$ with non-conditional $3$-transformers and the
rules of the two kinds above, that is, 
rules of type $\exists X$ and $\exists(k_1,\dots , k_n)$.
As already done in \cite{gamesandcomputation}, we also 
study constructions. Consider the class $\mathcal{C}$ of all finite $\tau$-models. We say
that $\mathcal{R}\subseteq \mathcal{C}
\times\mathcal{C}$ is an $\mathrm{RE}$-construction if there 
exists a possibly nondeterministic Turing machine $\mathit{TM}$ 
such that the following holds: we have $(\mathfrak{M},\mathfrak{N})\in \mathcal{R}$ iff with some
input $\mathit{enc}_<(\mathfrak{M})$ for some $<$
there exists some computation such that
$\mathit{TM}$ halts in an 
accepting state such that we have $\mathit{enc}_{<'}(\mathfrak{N})$ for
some $<'$ on the output tape at halting.
Note that there are many natural equivalent formulations of the notion.
We say that $\mathrm{NRL}$ can
\textbf{compute} $\mathcal{R}\subseteq \mathcal{C}\times\mathcal{C}$ if
there exists a program $\Pi$ such that the 
following holds: the program $\Pi$ can halt on the input $\mathfrak{M}\in \mathcal{C}$ 
with $X_{\mathit{true}}$ holding and 
the current $\tau$-model at halting being $\mathfrak{N}\in \mathcal{C}$ iff we
have $(\mathfrak{M},\mathfrak{N})\in \mathcal{R}$.
Note that ``can halt'' here of course means that there exists a 
favourable computation under the different possibilities allowed by the available nondeterminism.
Note also that we do not care about tape predicates when considering 
what the input and output models are: the models $\mathfrak{M}$ and $\mathfrak{N}$ are $\tau$-models,
even though during computation we modify $\tau^+$-models that take into accout tape predicates.

The following is a rather trivial 
variant of Theorem \ref{freeofordering}, now with the
nondeterministic logic $\mathrm{NRL}$.
In the following, $\mathcal{C}$ is the class of all finite $\tau$-models
where $\tau$ is any finite relational vocabulary.

\begin{theorem}\label{constone}
$\mathrm{NRL}$ can
compute $\mathcal{R}\subseteq \mathcal{C}\times\mathcal{C}$
iff $\mathcal{R}$ is an $\mathrm{RE}$-construction. 
\end{theorem}

\begin{proof}
This is a trivial variant of Theorem \ref{freeofordering}.
Consider $\tau$-models.
Note that $\mathcal{R}$ is an $\mathrm{RE}$-construction iff
there is a Turing machine that recognizes (in $\mathrm{RE}$) the class $\mathcal{D}$ of $(\tau\cup \{\sim\}\cup {P})$-models 
defined as follows.
\begin{enumerate}
\item
All models in $\mathcal{D}$ consist of a disjoint union of
two models $\mathfrak{A}\in \mathcal{C}$ and $\mathfrak{B} \in \mathcal{C}$.
The additional binary relation $\sim$ is an equivalence relation with two equivalence
classes, one class consisting of the domain of $\mathfrak{A}$ and the
other one of the domain of $\mathfrak{B}$. The unary predicate $P$ is 
true on precisely the elements belonging to one (but not the other) equivalence class.
(Intuitively, the unary predicate $P$ denotes which model is the input model.) Thus $\mathcal{D}$
encodes the pairs $(\mathfrak{C},\mathfrak{D})$ where $P$ denotes the first
member $\mathfrak{C}$ of the pair.
\item
We have $(\mathfrak{A},\mathfrak{B})\in \mathcal{R}$ iff $\mathcal{D}$ contains a
model that represents $(\mathfrak{A},\mathfrak{B})$.
\end{enumerate}
To compute $\mathcal{R}$ with $\mathrm{NLR}$, we write a program that first creates,
when the input model is $\mathfrak{A}$, some model that encodes a 
pair $(\mathfrak{A},\mathfrak{B})$ (see the above description). This construction is 
done non-deterministically, with the possibility to construct any finite $\tau$-model
whatsoever to represent $\mathfrak{B}$.
Then we use Theorem \ref{freeofordering} and the fragment $\mathrm{RL}$ to 
recognize the model class $\mathcal{D}$ as described above, that is, our program then 
makes true the predicate $X_{\mathit{true}}$ and ``halts'' iff the
model encoding $(\mathfrak{A},\mathfrak{B})$ is in $\mathcal{D}$. Note
that ``halts'' here means that we first make an auxiliary predicate $X_{halt}$ true and
then continue the computation as follows.\footnote{We could even require that $X_{true}$ is $X_{halt}$.}
When $X_{halt}$ has become true, we
make sure that only $\mathfrak{B}$ remains as an output model (the tape 
predicates used in the compututation do not count, only the relation symbols in $\tau$). 
For this we use the deletion operator as one of the constructs. After this we halt.
\end{proof}

It is also easy to show that $\mathrm{RL}$ with $3$-transformers added can compute a partial
function $\mathcal{R}\subseteq \mathcal{C}\times\mathcal{C}$ iff $\mathcal{R}$ is a
partial function that is recursively enumerable such that there is Turing
machine for the partial function $\mathcal{R}$. This latter condition 
means that there is a deterministic Turing machine $\mathit{TM}$ such that
given any $\mathfrak{A}\in \mathcal{C}$ 
and any $<$, the machine $\mathit{TM}$ halts on the input $\mathit{enc}_<(\mathfrak{A})$ iff
there exists a model $\mathfrak{B}$ such that $(\mathfrak{A},\mathfrak{B})\in \mathcal{R}$, and
furthermore, the output on halting is then $\mathit{enc}_{<'}(\mathfrak{B})$ for some $<'$.

Now, the statement corresponding to 
Theorem \ref{constone} holds also for the Turing-complete logic $\mathcal{L}[;]$ as
defined in \cite{gamesandcomputation}. We write $(\mathfrak{M},\varphi,\mathfrak{N})$ if
Eloise has a winning strategy in the semantic game beginning with $\mathfrak{M}$ and $\varphi$ and
with Eloise being the verifier and with the general assignment function being empty, and furthermore, that winning strategy always leads to Eloise winning so
that the current model at the time of 
winning is $\mathfrak{N}$. Note that tape predicates (encoded in the general assignment) and first-order variables (also in the general assignment) do not count towards what the 
final model $\mathfrak{N}$ looks like, only the relations in the
signature $\tau$ of the models $\mathfrak{M}$ and $\mathfrak{N}$ count.
The winning strategy can be assumed positional (this clearly makes no difference due to the
positional determinacy of reachability games even on infinite arenas). The winning
strategy could even be
assumed finite (which is typical in the logic $\mathcal{L}[;]$ and its
relatives in \cite{turingcomp}), as by K\"{o}nig's lemma, in every semantic game where
Eloise has a winning strategy, we can
find a bound on how many rounds a play can last before a win occurs.
This follows due to the game-tree
being finitely branching.

We say $\mathcal{L}[;]$ can 
compute $\mathcal{R}$ if there is a formula $\varphi$ of $\mathcal{L}[;]$ such that we
have $(\mathfrak{M},\varphi,\mathfrak{N})$ iff $(\mathfrak{M},\mathfrak{N}) \in \mathcal{R}$.

The following has a similar counterpart proven already in \cite{gamesandcomputation}.
Here $\mathcal{C}$ is the class of all finite $\tau$-models for any finite
relational vocabulary $\tau$.

\begin{theorem}\label{constone2}
$\mathcal{L}[;]$ can
compute $\mathcal{R}\subseteq\mathcal{C}\times\mathcal{C}$ iff $\mathcal{R}$ is an $\mathrm{RE}$-construction.

\end{theorem}
\begin{proof}
The proof is almost identical to the proof of the above Theorem. Moreover, the 
technicalities of the argument have
essentially already been given in \cite{gamesandcomputation}.

So, firstly, simulating $\mathcal{L}[;]$ by an alternating Turing-machine is straightforward,
and we can turn this simulation so that it runs with a nondeterministic machine of course.
For the other direction, simulating a Turing-machine 
with $\mathcal{L}[;]$, we write a formula $\varphi$ of $\mathcal{L}[;]$ to do as
described in the proof of Theorem \ref{constone}. Indeed, suppose an input model $\mathfrak{A}$ is
given. The formula lets Eloise construct the model $\mathfrak{S}$ corresponding to $(\mathfrak{A},\mathfrak{B})$ (here anything can potentially be constructed as model $\mathfrak{B}$ by 
Eloise). Then the formula allows Eloise to enter to
play the game with Ablelard to check
whether $\mathfrak{S}$ belongs to $\mathcal{D}$. This part can be done by 
the Turing-completeness of $\mathcal{L}[;]$. After winning 
this game, Eloise simply should still make
sure, using deletion operators, that the output model is $\mathfrak{B}$.
Tape predicates and first-order variables do not count towards what the output (or input) model is.

In the above construction, the composition connective $;$ is used to make sure
Abelard cannot end the game by losing at
some early stage before the constructions are ready. To put this shortly, Abelard
cannot stop the constructions by losing the game intentionally too early in the play.
\end{proof}

\subsection{A general setting}

Consider the following transform rules (with $\varphi$ and $\psi(x)$ 
allowed to be written in existential second-order logic).
\begin{enumerate}
\item
$X(x_1,\dots , x_k)\ :-\ \ \varphi$
\item
$I$
\item
$D\ :-\ \ \psi(x)$
\item
$\exists X$
\end{enumerate}
A \textbf{conditional transformer tuple} is a conditional 
rule tuple 
$$(\mathit{If}\, \varphi_1\, \mathit{then}\, F_1,\, 
\mathit{else\, if}\, \varphi_2\, \mathit{then}\, F_2,\ \dots\ ,\,  
\mathit{else\, if}\, \varphi_{k-1}\, \mathit{then}\, F_{k-1},\ \mathit{else}\, F_k)$$
where each rule $F_i$ is a rule of the above four kinds listed, that is, of type
\begin{enumerate}
\item
$X(x_1,\dots , x_k)\ :-\ \ \varphi$
\item
$I$ 
\item
$D\ :-\ \ \psi(x)$
\item
$\exists X$
\end{enumerate}
where we allow the formulae $\varphi$ and $\psi(x)$ to be 
formulae of existential second-order logic.
Furthermore, any of the sentences $\varphi_i$ in the conditional rule 
tuple can be a sentence of existential second-order logic.
The case $k=1$ is allowed, and then the conditional
transformer tuple is just a rule $F_1$.
A \textbf{non-deterministic transformer} is a tuple
$$(C_1,\dots , C_m)$$
where each $C_i$ is a conditional transformer tuple.
The idea is to consider $(C_1,\dots , C_m)$ as a rule
such that when executed, we nondeterministically 
pick one $C_i$ and execute it. A \textbf{parallel transformer} is a
tuple $(T_1,\dots , T_n)$ where each $T_i$ is a non-deterministic
transformer. These are executed as follows. Beginning
from $$(T_1,\dots , T_n) = \bigl((C_{1,1},\dots , C_{1,m_1}),
\dots , (C_{n,1},\dots , C_{n,m_n})\bigr),$$ we get as output, using  
non-determinism (which can be guided by $n$ different agents or $n$ 
deterministic strategies, possibly 
encoded by $n$ automata), a
tuple $(C_{1,j_1},\dots , C_{n,j_n})$ of conditional transformer tuples. 
The tuple is obtained 
such that each agent $k\in \{1,\dots , n\}$
chooses $C_{k,j_k}$ from $(C_{k,1},\dots , C_{k,m_k})$. The output
tuple $(C_{1,j_1},\dots , C_{n,j_n})$
turns into a tuple $$(F_1,\dots , F_n)$$ of
rules, where $F_k$ is determined based on $C_{k,j_k}$ and its
internal structure. Note that we are thus playing a game where
each agent $k$ makes the nondeterministic choice from $(C_{k,1},\dots , C_{k,m_k})$.
The obtained rule tuple $(F_1, \dots , F_k)$ is then executed in parallel as 
described above when defining the way parallel rules are treated.
Note that if a deadlock is obtained, the computation ends without the
current model being modified. A parallel
transformer is considered a single line of code. It relates to the agents 
making a parallel choice.\footnote{Note that we can guide the computation 
line flow with parallel transformers as well if we (1) encode 
nullary predicates to be modified in the parallel transformer and then (2) 
write further rules (in subsequent lines) that choose the outcome line to be executed based on the 
results of the transformation. Of course we can even define a general parallel rule
where we can directly obtain also control rules, not only transform rules.}

Let a \textbf{conditional flow control rule} be a
tuple of the form 
$$(\mathit{If}\, \varphi_1\, \mathit{then}\, \mathbf{p}_1,\, 
\mathit{else\, if}\, \varphi_2\, \mathit{then}\, \mathbf{p}_2,\ \dots\ ,\,  
\mathit{else\, if}\, \varphi_{k-1}\, \mathit{then}\, \mathbf{p}_{k-1},\ \mathit{else}\, 
\mathbf{p}_k)$$
where each $\mathbf{p}_i$ is a nondeterministic control rule of
type $\exists(k_{i_1},\dots , k_{i_{\ell}})$ as
defined above. The case $k=1$ is of course allowed, 
being the case where the flow 
control rule is just a single rule $\exists(k_{i_1},\dots , k_{i_{\ell}})$.

Let $\mathrm{GRL}$ denote the logic where we have all
parallel transformers and conditional flow control rules.
Let \textbf{sorted} $\mathrm{GRL}$ be the logic 
where each line in each program has one of two labels: $A$ or $G$, for 
\emph{agents} and \emph{general controller}. 
These lines are also called $A$-lines and $G$-lines. The $A$-lines are 
parallel transformers $(T_1,\dots , T_m)$ where $m$ is the same for
each $A$-line of the program (the number of agents). $G$-lines are conditional flow
control rules or lines of  the following types. 
\begin{enumerate}
\item
$X(x_1,\dots , x_k)\ :-\ \ \varphi$
\item
$I$ 
\item
$D\ :-\ \ \psi(x)$
\item
$\exists X$
\end{enumerate}
where all the formulae can be in existential second-order logic.

The point of sorted $\mathrm{GRL}$ is
that we guide systems as defined in the beginning of the article. 
The parallel transformers are guided by agents, $(T_1,\dots , T_m)$ being a tuple for $m$ agents.
The other rules are controlled by the general controller $G$. What the agents are trying to
achieve can be specified in many ways, depending on the
modelling purpose. However, one scenario is that the 
agents are jointly trying to make the system halt with $X_{true}$ holding. We say that the agents have a winning strategy with
the input model $\mathfrak{M}$ if there exist 
functions $f_1,\dots , f_m$ that give the choices for nondeterminism in 
parallel transform rules in computations beginning with $\mathfrak{M}$.
When the functions $f_1,\dots , f_m$ are followed, then every computation
leads to the system halting with $X_{true}$ holding. However, this is just a reachability
game. Many other settings are interesting. It is also of utmost
interest to limit the domains of $f_1,\dots , f_m$. For example, we could 
make each $f_i$ depend only on some single predicate $R_i$, conceived as the 
range of (physical or even perhaps mental)
perception (or horizon) of agent $i$.

\bibliographystyle{plain}
\bibliography{mybib}


\end{document}